\documentclass[conference]{IEEEtran}
\IEEEoverridecommandlockouts

\usepackage{amsmath,amssymb,amsfonts}
\usepackage{amsthm}
\usepackage{algorithm}
\usepackage{algpseudocode}
\usepackage{graphicx}
\usepackage{tikz}
\usepackage{textcomp}
\usepackage{xcolor}
\usepackage{bm}
\usepackage{cite}

\newcommand{\msgset}{\mathcal{M}}                 
\newcommand{\nmsg}{M}                             
\newcommand{\rxseq}{\mathbf{y}}                   
\newcommand{\rxupto}{\mathbf{y}^{t}}              
\newcommand{\predset}[1]{\Gamma(#1)}              
\newcommand{\cps}{\mathcal{T}}                    
\newcommand{\dline}{D}                            
\newcommand{\sstop}{T_{\mathrm{s}}}               
\newcommand{\epstarg}{\varepsilon_{\mathrm{tar}}} 
\newcommand{\dcal}{\mathcal{D}_{\mathrm{cal}}}    
\newcommand{\thr}[1]{s_{\mathrm{th}}^{#1}}        
\newcommand{\correct}{m}                          
\newcommand{\decided}{\hat{m}}                    

\newtheorem{proposition}{Proposition}

\begin{document}

\title{Conformal Decode-or-Erase: Certified Spiking Decoding for Short-Packet URLLC}

\author{Zihang~Song, Kai Yu, Anders E. Kalør and Petar Popovski
\thanks{The work was supported by Horizon Europe Marie Skłodowska-Curie Action (MSCA) Postdoc Fellowships project NEUER with grant No. 101206861, and the Velux Foundation, Denmark, through the Villum Investigator Grant WATER, nr. 37793.}
\thanks{Z. Song (zsong@es.aau.dk) and P. Popovski (petarp@es.aau.dk) are with the Connectivity Section, Department of Electronic Systems, Aalborg University, 9220 Aalborg, Denmark. K. Yu (kayu@kth.se) is with the School of Electrical Engineering and Computer Science, KTH Royal Institute of Technology, SE-100 44 Stockholm, Sweden. A. E. Kalør (aek@keio.jp) is with the Department of Information and Computer Science, Keio University, Yokohama 223-8522, Japan.}
}

\maketitle

\begin{abstract}
Ultra-reliable low-latency communication (URLLC) must deliver short packets within a hard deadline at low error probability. A conventional receiver waits for the full packet before deciding, spending the full latency and energy even though many packets are resolvable well before the deadline. Committing early without a reliability guarantee, however, risks a silent wrong delivery, so the receiver is left choosing between wasted resources and uncontrolled errors. We propose Conformal Decode-or-Erase (CoDE), a spiking neural network (SNN) receiver that resolves this tension. The SNN reads one symbol per channel use and forms, at predetermined checkpoints, a set of candidate messages that provably contains the true one with a prescribed probability. CoDE commits once the set narrows to a singleton and otherwise declares an erasure that triggers hybrid automatic repeat request (HARQ) retransmission. A wrong commit means the true message fell outside that singleton. Hence, the prediction set provides an upper bound on the undetected error rate in a distribution-free manner and for any pretrained SNN and any calibration size. Simulations confirm reliability at roughly half a fixed-length decoder's latency and compute.
\end{abstract}


\section{Introduction}
\label{sec:intro}

Ultra-reliable low-latency communication (URLLC) underpins factory automation, vehicular control, and other mission-critical services, where a packet must be delivered within a hard latency budget and at a very low error probability~\cite{durisi2016}. These services rely on short packets, for which the finite-blocklength penalty is severe~\cite{polyanskiy2010}. How many channel uses a packet needs before it can be decoded varies with its difficulty and the signal-to-noise ratio (SNR), and many packets are resolvable well before the deadline. Conventional receivers do not exploit this: a classical decoder validates the full packet with a cyclic redundancy check (CRC), and a dense neural receiver such as the fully convolutional DeepRx~\cite{deeprx} runs the same full-packet forward pass on every packet. Both spend the whole latency budget and its energy regardless of packet difficulty, wasting both on the easy ones, a fixed cost the energy-constrained machine-type and Internet-of-Things endpoints of URLLC can least afford.

Spiking neural networks (SNNs), the basis of neuromorphic computing, remove this waste by processing the packet as a stream of events: they update one symbol at a time, so a decision can be reached before the full packet arrives, and their energy scales with the spikes they emit rather than with the packet length~\cite{neftci2019,snntorch}. Multiplication-free spikes replace the multiply-accumulate operations of a dense neural receiver, and SNN receivers have shown substantial energy reductions~\cite{neurocomm,song2024,song2024xpikeformer,spikingrx,spikacom,song2025leo}.

%
%

\usetikzlibrary{arrows.meta,positioning,calc,backgrounds,shapes.symbols}

\definecolor{ink}{HTML}{1A1A1A}
\definecolor{commit}{HTML}{1F7A53}
\definecolor{erase}{HTML}{C0392B}

\tikzset{
  block/.style={rounded corners=2.5pt, draw=ink!55, line width=0.7pt, fill=white},
  btitle/.style={font=\footnotesize\bfseries, text=ink},
  sub/.style={font=\footnotesize, text=ink!65},
  msgbox/.style={rounded corners=2pt, draw=ink!70, fill=ink!7, font=\bfseries, text=ink, inner sep=2.4pt},
  sym/.style={draw=ink!60, fill=ink!10, line width=0.5pt},
  neuron/.style={circle, minimum size=4.3mm, inner sep=0pt, draw=ink!65, fill=white, line width=0.6pt},
  fneuron/.style={circle, minimum size=4.3mm, inner sep=0pt, draw=ink!75, fill=ink!58, line width=0.6pt},
  syn/.style={draw=ink!18, line width=0.32pt},
  flow/.style={-{Latex[length=2.2mm,width=1.7mm]}, line width=0.9pt, draw=ink!80},
  spike/.style={ink!72, line width=0.8pt, line cap=round},
  inset/.style={rounded corners=1pt, minimum width=5mm, minimum height=3.6mm, inner sep=0pt,
                draw=commit, fill=commit!16, text=commit!45!black, font=\footnotesize\bfseries},
  exset/.style={rounded corners=1pt, minimum width=5mm, minimum height=3.6mm, inner sep=0pt,
                draw=ink!30, fill=ink!4, text=ink!40, font=\footnotesize},
  erset/.style={rounded corners=1pt, minimum width=5mm, minimum height=3.6mm, inner sep=0pt,
                draw=erase!55, fill=erase!12, text=erase!60!black, font=\footnotesize\bfseries},
  outbox/.style={rounded corners=2.5pt, line width=0.9pt, font=\bfseries, inner sep=4pt},
}
\providecommand{\setstack}[3]{
  \foreach \mb [count=\i from 1] in {#3}{%
    \ifnum\mb=1 \node[inset] at (#1,#2-\i*0.40) {\i};%
    \else \node[exset] at (#1,#2-\i*0.40) {\i};\fi}}
\providecommand{\setstackE}[3]{
  \foreach \mb [count=\i from 1] in {#3}{%
    \ifnum\mb=1 \node[erset] at (#1,#2-\i*0.40) {\i};%
    \else \node[exset] at (#1,#2-\i*0.40) {\i};\fi}}
\providecommand{\antenna}[3]{%
  \draw[#3, line width=0.9pt] (#1,#2) -- (#1,#2+0.42);
  \draw[#3, line width=0.9pt] (#1,#2+0.42) -- ++(-0.13,0.22);
  \draw[#3, line width=0.9pt] (#1,#2+0.42) -- ++(0.13,0.22);
  \draw[#3, line width=0.8pt] (#1-0.12,#2) -- (#1+0.12,#2);}

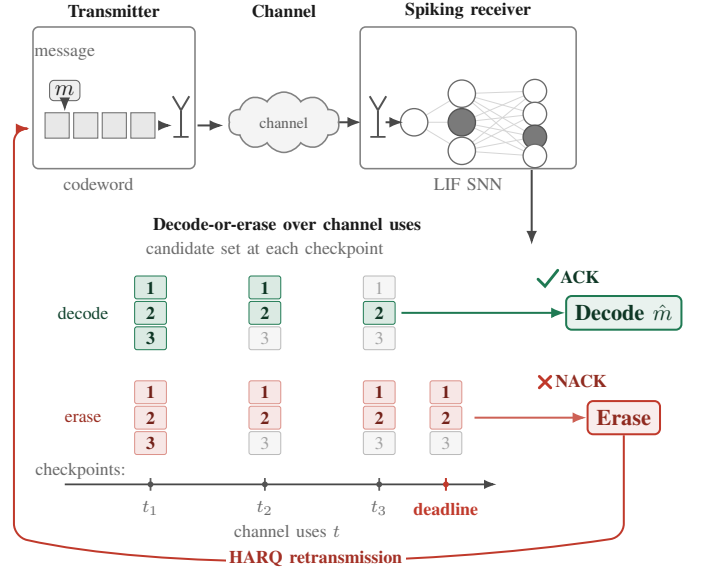
\begin{figure}[t]
\centering
\resizebox{\columnwidth}{!}{%
\begin{tikzpicture}[font=\sffamily]

\node[btitle] at (1.82,8.30) {Transmitter};
\node[btitle] at (4.50,8.30) {Channel};
\node[btitle] at (7.40,8.30) {Spiking receiver};

\draw[block] (0.55,5.82) rectangle (3.10,8.05);
\node[sub] at (1.05,7.62) {message};
\node[msgbox] (msg) at (1.05,7.05) {$\correct$};
\foreach \j in {0,1,2,3}{ \draw[sym] (0.74+\j*0.45,6.30) rectangle ++(0.38,0.40); }
\draw[flow] (msg.south) -- (1.05,6.72);
\antenna{2.88}{6.32}{ink!75}
\draw[flow] (2.56,6.50) -- (2.74,6.50);
\node[sub] at (1.58,5.58) {codeword};
\draw[flow] (3.14,6.50) -- (3.55,6.50);

\node[cloud, draw=ink!55, fill=ink!5, line width=0.7pt,
      cloud puffs=10, cloud puff arc=120, aspect=2.2,
      minimum width=1.75cm, minimum height=0.95cm,
      font=\scriptsize, text=ink!70] (ch) at (4.50,6.55) {channel};
\draw[flow] (ch.east) -- (5.72,6.55);

\draw[block] (5.70,5.82) rectangle (9.10,8.05);
\antenna{5.98}{6.32}{ink!75}
\coordinate (in) at (6.55,6.55);
\coordinate (h1) at (7.30,7.00); \coordinate (h2) at (7.30,6.55); \coordinate (h3) at (7.30,6.10);
\coordinate (o1) at (8.45,7.04); \coordinate (o2) at (8.45,6.68); \coordinate (o3) at (8.45,6.32); \coordinate (o4) at (8.45,6.02);
\foreach \h in {h1,h2,h3}{ \draw[syn] (in) -- (\h); }
\foreach \h in {h1,h2,h3}{ \foreach \o in {o1,o2,o3,o4}{ \draw[syn] (\h) -- (\o); }}
\node[neuron] at (in){};
\node[neuron] at (h1){}; \node[fneuron] at (h2){}; \node[neuron] at (h3){};
\node[neuron, minimum size=3.7mm] at (o1){}; \node[neuron, minimum size=3.7mm] at (o2){};
\node[fneuron, minimum size=3.7mm] at (o3){}; \node[neuron, minimum size=3.7mm] at (o4){};
\draw[flow] (6.10,6.55) -- (6.38,6.55);
\node[sub] at (7.40,5.58) {LIF SNN};
\draw[flow] (8.40,5.80) -- (8.40,4.62);   

\node[font=\footnotesize\bfseries, text=ink] at (4.55,4.95) {Decode-or-erase over channel uses};
\node[sub] at (4.20,4.54) {candidate set at each checkpoint};

\node[sub, text=commit!55!black] at (1.34,3.55) {decode};
\setstack{2.40}{4.35}{1,1,1}
\setstack{4.20}{4.35}{1,1,0}
\setstack{6.00}{4.35}{0,1,0}

\node[sub, text=erase!80!black] at (1.34,1.90) {erase};
\setstackE{2.40}{2.70}{1,1,1}
\setstackE{4.20}{2.70}{1,1,0}
\setstackE{6.00}{2.70}{1,1,0}
\setstackE{7.05}{2.70}{1,1,0}

\draw[-{Latex[length=2mm]}, line width=0.8pt, ink!80] (1.05,0.85) -- (7.85,0.85);
\node[sub] at (4.55,0.12) {channel uses $t$};
\node[sub] at (1.30,1.06) {checkpoints:};
\foreach \x/\lab in {2.40/{$t_1$}, 4.20/{$t_2$}, 6.00/{$t_3$}}{
  \draw[ink!70, line width=0.8pt] (\x,0.75) -- (\x,0.95);
  \fill[ink!75] (\x,0.85) circle (1.3pt);
  \node[sub] at (\x,0.48) {\lab};}
\draw[erase, line width=0.9pt] (7.05,0.75) -- (7.05,0.95);
\fill[erase] (7.05,0.85) circle (1.3pt);
\node[font=\footnotesize\bfseries, text=erase] at (7.05,0.46) {deadline};

\node[outbox, draw=commit, fill=commit!10, text=commit!45!black] (dec) at (9.85,3.55) {Decode $\decided$};
\draw[flow, draw=commit] (6.35,3.55) to[out=0,in=180] (dec.west);
\draw[commit, line width=1.3pt, line cap=round] (8.50,4.08)--(8.61,3.97)--(8.81,4.22);
\node[font=\footnotesize\bfseries, text=commit!45!black] at (9.16,4.12) {ACK};

\node[outbox, draw=erase, fill=erase!9, text=erase!75!black] (era) at (9.85,1.90) {Erase};
\draw[flow, draw=erase!80] (7.50,1.90) to[out=0,in=180] (era.west);
\draw[erase, line width=1.3pt, line cap=round] (8.52,2.37)--(8.70,2.55) (8.52,2.55)--(8.70,2.37);
\node[font=\footnotesize\bfseries, text=erase!75!black] at (9.18,2.47) {NACK};

\draw[erase, line width=0.9pt, -{Latex[length=2.2mm]}, rounded corners=6pt]
  (era.south) -- (9.85,0.02) .. controls (6.5,-0.48) and (2.5,-0.48) .. (0.25,0.02)
  -- (0.25,6.45) -- (0.50,6.45);
\node[font=\footnotesize\bfseries, text=erase!85!black, fill=white, inner sep=1.5pt]
  at (5.00,-0.32) {HARQ retransmission};

\end{tikzpicture}}
\caption{CoDE overview. The spiking receiver reads one symbol per channel use and narrows a conformal set of candidate messages at each checkpoint. It commits when the set reduces to one message and returns an ACK. When no checkpoint up to the deadline yields a singleton, the receiver cannot commit, so it erases and returns a NACK, which triggers a HARQ retransmission.}
\label{fig:system}
\end{figure}

However, deciding early is unsafe without a guarantee on the undetected error rate, which neither approach gives: even the optimal maximum-likelihood (ML) decoder returns only a hard decision, and a CRC only fixes a residual undetected-error floor set by the code rather than a target the receiver sets. Conformal prediction (CP) supplies the missing certificate, distribution-free and finite-sample~\cite{vovk}, and has reached wireless calibration, scheduling, and demodulation~\cite{cohen2023tmlcn,cohen2023spl,cohen2023demod}. Closest to our setting, SpikeCP wraps a pretrained SNN in CP and stops once the prediction set is small enough to meet a target reliability~\cite{spikecp}, but it bounds set coverage, not the undetected error rate URLLC is defined over, and it targets vision-domain classification, not a physical-layer decoder with an erasure-to-HARQ mapping, a hard latency budget, and airtime energy.
We close this gap with Conformal Decode-or-Erase (CoDE), a neuromorphic receiver that decides when to decode (Fig.~\ref{fig:system}). It supplies a decoder-agnostic, distribution-free, finite-sample certificate on the undetected error rate and an adaptive per-packet stop; the certificate follows from the coverage guarantee alone, so it applies to any pretrained decoder, and the SNN is the low-energy example we evaluate. The checkpoints are placed within the latency budget so that the deadline is itself a checkpoint and any unresolved packet is erased in time for an HARQ round. At each checkpoint, CoDE forms a set of candidate messages that provably contains the true one with prescribed probability, commits when that set is a singleton, and otherwise declares an erasure that triggers an HARQ retransmission. A commit to a wrong message requires the true message to fall outside the set, so the conformal coverage level upper-bounds the undetected error rate, independent of the SNN quality and of the calibration size. The channel uses consumed before a commit set both the latency and the airtime energy, while the internal spike count sets the compute energy; one rule therefore governs reliability, latency, and energy together. A deadline-aware budget allocation relaxes the per-checkpoint requirement toward the deadline to lower the erasure rate at the same guarantee.

To our knowledge, CoDE is the first physical-layer receiver to certify this bound while adapting latency and energy to each packet. We make the following contributions.
\begin{itemize}
\item We formulate CoDE, a conformal decode-or-erase rule for a spiking receiver, and show that the conformal coverage level upper-bounds the undetected error rate.
\item We map the SNN time axis onto channel uses, so a single stopping rule sets reliability, latency, and energy, where energy includes both spikes and airtime.
\item We introduce a deadline-aware allocation of the per-checkpoint error budget that lowers the erasure rate relative to a uniform split and preserves the guarantee.
\item We implement and evaluate CoDE in Sionna on the additive white Gaussian noise (AWGN) channel against fixed-length, coverage-only, and non-spiking NN baselines and a standard ML+CRC+HARQ receiver.
\end{itemize}

Section~\ref{sec:model} states the system model and metrics; Section~\ref{sec:method} presents CoDE and its guarantee; Section~\ref{sec:exp} reports the experiments; Section~\ref{sec:concl} concludes.
\section{System Model}
\label{sec:model}

\subsection{Short-Packet Transmission and Channel}
\label{sec:model-tx}
A transmitter sends one of $\nmsg$ equiprobable messages $\correct \in \msgset = \{1,\dots,\nmsg\}$ over a short packet. Message $\correct$ is mapped to a length-$\dline$ symbol sequence $\mathbf{x}(\correct) = (x_1(\correct),\dots,x_{\dline}(\correct))$, with $x_t(\correct) \in \mathbb{C}$ drawn from a short codebook, where $\dline$ is the latency budget in channel uses, not the uncoded length $\lceil\log_2\nmsg/Q\rceil$ for a modulation order $Q$; the codebook itself is the code.

The receiver observes $y_t$ at channel use $t$ over the AWGN channel,
\begin{equation}
y_t = x_t(\correct) + n_t, \qquad n_t \sim \mathcal{CN}(0, N_0).
\label{eq:awgn}
\end{equation}
We write $\rxupto = (y_1,\dots,y_t)$ for the received samples up to channel use $t$, and $\rxseq$ for the full packet.

\subsection{Spiking Receiver}
\label{sec:model-rx}
The receiver uses an SNN with leaky integrate-and-fire (LIF) neurons~\cite{neftci2019,snntorch}. At each channel use $t$ it takes the real and imaginary parts of $y_t$ as input, updates its internal membrane states, and drives $\nmsg$ readout neurons, one per candidate message. Under rate decoding, readout neuron $m$ accumulates the spike count $r_m(\rxupto) = \sum_{t'=1}^{t} b_{m,t'}$, where $b_{m,t'} \in \{0,1\}$ is its spike at channel use $t'$, a causal function of the inputs $\mathbf{y}^{t'}$ received up to $t'$. A softmax over the count vector yields the predictive probabilities
\begin{equation}
p_m(\rxupto) = \frac{e^{r_m(\rxupto)}}{\sum_{m'=1}^{\nmsg} e^{r_{m'}(\rxupto)}}.
\label{eq:prob}
\end{equation}

We assume a calibration set $\dcal = \{(\rxseq[i], \correct[i])\}_{i=1}^{|\dcal|}$ of labeled packets exchangeable with the test packet. In practice these are pilot packets or CRC-verified decodes. The calibration size sets a floor on the achievable target: the guarantee requires $\epstarg \ge 1/(|\dcal|+1)$, so reaching operational URLLC levels ($\epstarg \lesssim 10^{-5}$) demands either a large calibration set or extreme-quantile methods.

\subsection{Decision Timing and HARQ}
\label{sec:model-timing}
At each checkpoint $t \in \cps$, where $\cps \subseteq \{1,\dots,\dline\}$ with $\dline \in \cps$, the receiver forms a prediction set $\predset{\rxupto} \subseteq \msgset$, constructed in Section~\ref{sec:method}. If any checkpoint yields a singleton, the receiver commits its element $\decided$ at the earliest such checkpoint,
\begin{equation}
\sstop(\rxseq) = \min\{t \in \cps : |\predset{\rxupto}| = 1\},
\label{eq:stop}
\end{equation}
and sends an acknowledgment (ACK). Otherwise it declares an erasure at $\dline$, and the negative acknowledgment (NACK) triggers an HARQ retransmission.

\subsection{Performance Metrics}
\label{sec:model-metrics}
We evaluate the following quantities, with all probabilities taken over the test packet, the channel, and the calibration set.

\emph{Reliability.} The undetected error rate is the probability of committing to a wrong message,
\begin{equation}
P_{\mathrm{ue}} = \Pr\!\big(\decided \neq \correct,\ \text{commit}\big),
\label{eq:ue}
\end{equation}
and the receiver is reliable at target $\epstarg$ if $P_{\mathrm{ue}} \le \epstarg$. An erasure is not an error; it costs a retransmission rather than a silent wrong delivery.

\emph{Latency.} The decision latency is the average stopping time $\mathbb{E}[\sstop(\rxseq)]$ in channel uses. The erasure rate is the probability that no checkpoint yields a singleton,
\begin{equation}
P_{\mathrm{er}} = \Pr\!\big(|\predset{\rxupto}| \neq 1 \ \text{for all}\ t \in \cps\big).
\label{eq:per}
\end{equation}
Under HARQ, the delivered latency compounds the stopping time across retransmission rounds.

\emph{Energy.} The compute energy is proportional to the expected internal spike count $\mathbb{E}[S(\rxseq)]$, and the airtime energy is proportional to the expected channel uses $\mathbb{E}[\sstop(\rxseq)]$.

\section{Conformal Decode-or-Erase}
\label{sec:method}

\subsection{Conformal Set Construction}
\label{sec:method-set}
At a checkpoint $t \in \cps$, a nonconformity (NC) score $s_m(\rxupto)$ measures the lack of confidence of the SNN in message $m$. We use the global log-loss score
\begin{equation}
s_m(\rxupto) = -\log p_m(\rxupto),
\label{eq:nc-global}
\end{equation}
which couples all readout neurons. The prediction set collects the messages whose score is below a checkpoint threshold $\thr{t}$,
\begin{equation}
\predset{\rxupto} = \{\, m \in \msgset : s_m(\rxupto) \le \thr{t} \,\}.
\label{eq:set}
\end{equation}

The threshold follows split conformal prediction~\cite{vovk,tibshirani2019}. For each checkpoint $t$ and each calibration packet $i$, let $s^t[i] = s_{\correct[i]}(\mathbf{y}^{t}[i])$ be the NC score of its true message. Given a per-checkpoint budget $\alpha_t \in (0,1)$, set
\begin{equation}
\thr{t} = \big\lceil (1-\alpha_t)(|\dcal|+1) \big\rceil\text{-th smallest of } \{s^t[i]\}_{i=1}^{|\dcal|},
\label{eq:thr}
\end{equation}
when $\alpha_t \ge 1/(|\dcal|+1)$, and $\thr{t} = \infty$ otherwise. By the exchangeability of the calibration and test packets, this threshold yields the per-checkpoint coverage
\begin{equation}
\Pr\!\big(\correct \in \predset{\rxupto}\big) \ge 1 - \alpha_t.
\label{eq:cov}
\end{equation}

\subsection{Decode-or-Erase Rule and Reliability}
\label{sec:method-guarantee}
Recall the rule of Section~\ref{sec:model-timing}: CoDE commits the single element of the first singleton set at the stopping time~\eqref{eq:stop}, and erases at $\dline$ otherwise. The next result bounds its undetected error rate.

\begin{proposition}[Undetected-error certificate of CoDE]
\label{prop:ue}
Let the thresholds be set by~\eqref{eq:thr} with budgets $\{\alpha_t\}_{t\in\cps}$ that satisfy $\sum_{t\in\cps}\alpha_t = \epstarg$. Then the decode-or-erase rule has undetected error rate $P_{\mathrm{ue}} \le \epstarg$ for any pretrained SNN and any calibration size $|\dcal|$. We call this distribution-free, finite-sample bound the \emph{certificate} that CoDE attaches to each commit.
\end{proposition}
\begin{proof}
Let $\tau = \sstop(\rxseq)$ be the commit time. A commit gives $\predset{\mathbf{y}^{\tau}} = \{\decided\}$. If $\decided \neq \correct$, then $\correct \notin \predset{\mathbf{y}^{\tau}}$, so the true message is excluded from the prediction set at checkpoint $\tau \in \cps$. Therefore
\begin{equation}
\{\decided \neq \correct,\ \text{commit}\} \subseteq \bigcup_{t\in\cps}\{\,\correct \notin \predset{\rxupto}\,\}.
\label{eq:inclusion}
\end{equation}
The union bound and the per-checkpoint coverage~\eqref{eq:cov} give~\cite{spikecp}
\begin{equation}
P_{\mathrm{ue}} \le \sum_{t\in\cps}\Pr\!\big(\correct \notin \predset{\rxupto}\big) \le \sum_{t\in\cps}\alpha_t = \epstarg.
\end{equation}
\end{proof}
Unlike a conventional hard-decision receiver, which returns a decoded message with no such guarantee, CoDE certifies this bound for any decoder. Because CP acts only at inference, the guarantee is independent of how the SNN is trained, and the SNN quality sets only the achievable latency, energy, and erasure operating point. A CP-aware objective that targets the erasure rate at the deadline is a natural refinement.

\subsection{Error-Budget Allocation}
\label{sec:method-budget}
Proposition~\ref{prop:ue} holds for any budgets that sum to $\epstarg$, which leaves room for a design choice.

\emph{Uniform (Bonferroni).} The baseline splits the budget evenly, $\alpha_t = \epstarg/|\cps|$, as in SpikeCP. As $|\cps|$ grows, the per-checkpoint budget shrinks, the thresholds~\eqref{eq:thr} tighten, the sets~\eqref{eq:set} grow, and the erasure rate rises.

\emph{Deadline-aware.} We instead use weights $w_t \ge 0$ with $\sum_{t\in\cps} w_t = 1$, nondecreasing in the checkpoint rank $i_t = |\{t' \in \cps : t' \le t\}|$, and set $\alpha_t = w_t\,\epstarg$. A larger budget at the late checkpoints gives smaller sets there, so more packets commit by the deadline and fewer are erased. A representative instance is the linear allocation $w_t \propto i_t$. Since $\sum_t \alpha_t = \epstarg$ is the only condition Proposition~\ref{prop:ue} requires, any nondecreasing weight vector summing to one yields a valid allocation without additional assumptions.

\emph{Latency-erasure tradeoff.} The deadline-aware allocation trades a slightly stricter early checkpoint against a lower erasure rate on hard packets; the early checkpoint may delay a commit on easy packets, but for URLLC, where an erasure costs a full retransmission round, this exchange typically reduces the delivered latency.

\emph{Calibration floor.} The smallest achievable $\epstarg$ is limited by the calibration size: a per-checkpoint budget below $1/(|\dcal|+1)$ forces $\predset{\rxupto} = \msgset$, which never commits. Section~\ref{sec:exp} quantifies this floor.


\begin{algorithm}[t]
\caption{CoDE: conformal decode-or-erase}
\label{alg:code}
\begin{algorithmic}[1]
\Require trained SNN; calibration set $\dcal$; checkpoints $\cps$; target $\epstarg$; budget weights $\{w_t\}$
\Statex \textbf{Offline calibration}
\For{each checkpoint $t \in \cps$}
  \State $\alpha_t \gets w_t\,\epstarg$
  \State compute calibration scores $\{s^t[i]\}_{i=1}^{|\dcal|}$
  \State set $\thr{t}$ as the $(1-\alpha_t)$-quantile via~\eqref{eq:thr}
\EndFor
\Statex \textbf{Online decoding}
\For{$t = 1, 2, \dots, \dline$}
  \State update the SNN with received symbol $y_t$
  \If{$t \in \cps$}
    \State $\predset{\rxupto} \gets \{ m \in \msgset : s_m(\rxupto) \le \thr{t} \}$
    \If{$|\predset{\rxupto}| = 1$}
      \State commit $\decided \gets$ element of $\predset{\rxupto}$; send ACK
      \State \Return $\decided$ \Comment{stop, $\sstop = t$}
    \EndIf
  \EndIf
\EndFor
\State declare erasure; send NACK \Comment{HARQ retransmission}
\end{algorithmic}
\end{algorithm}

\providecommand{\todo}[1]{\textcolor{red}{\textbf{[TODO: #1]}}}
\providecommand{\incfig}[2]{\IfFileExists{#2}{\includegraphics[width=#1]{#2}}{\fbox{\parbox[c][2.6cm][c]{\dimexpr#1-2\fboxsep-2\fboxrule\relax}{\centering\ttfamily\scriptsize\color{black!55}[missing figure: \detokenize{#2}]}}}}

\section{Experiments}
\label{sec:exp}

\subsection{Setup}
\label{sec:exp-setup}
We implement CoDE in Sionna~\cite{sionna} for the transmit chain and the channel, and a from-scratch LIF SNN in snnTorch~\cite{snntorch} for the receiver, on a single GPU. We evaluate on the AWGN channel~\eqref{eq:awgn}, where the ML decoder is optimal and serves as our accuracy reference.
We set $\nmsg = 16$ messages on a length-$\dline = 32$ QPSK codebook and sweep $E_b/N_0$ from $-2$ to $10$\,dB with targets $\epstarg \in \{0.05, 0.10, 0.15, 0.20\}$. We adopt a fixed random codebook: each message maps to an i.i.d. QPSK sequence, the canonical finite-blocklength achievability construction~\cite{polyanskiy2010}, with no separate channel code, at rate $\log_2\nmsg/\dline = 0.125$ bit per channel use. We use a small $\nmsg$ and a simple fully-connected SNN architecture as a controlled proof of concept for the decode-or-erase mechanism. Since the guarantee of Proposition~\ref{prop:ue} is code-agnostic, larger codebooks can be handled by advanced streaming architectures~\cite{schoene2024,song2024xpikeformer}.

\begin{figure}[t]
\centering
\incfig{\columnwidth}{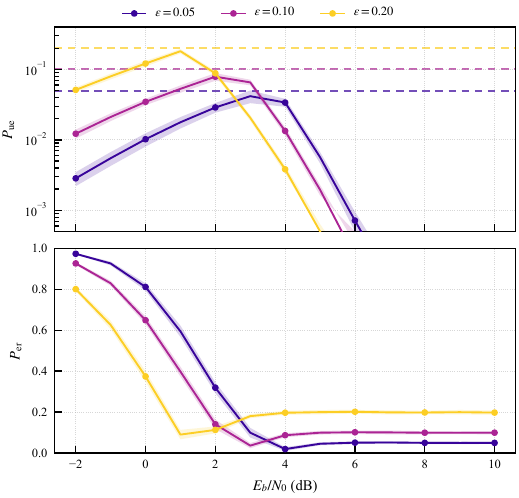}
\caption{Undetected error rate (top) and erasure rate (bottom) versus $E_b/N_0$ on AWGN, for $\epstarg \in \{0.05, 0.10, 0.15, 0.20\}$, $\nmsg=16$, $\dline=32$. Horizontal lines mark each target $\epstarg$. Bands span the calibration draws.}
\label{fig:reliability}
\end{figure}

The SNN is a two-layer LIF network with $256$ hidden neurons per layer, trained with surrogate gradients on the checkpoint cross-entropy with no conformal step at training time. At inference we apply CoDE with evenly spaced checkpoints $\cps$, the global log-loss score~\eqref{eq:nc-global}, and either the uniform Bonferroni or the deadline-aware budget allocation of Section~\ref{sec:method-budget}. The calibration set holds $|\dcal| = 2000$ packets; every operating point averages $25$ independent calibration draws against $2\times10^{4}$ test packets each. We fix $|\cps| = 8$ except in the allocation study, which sweeps $|\cps| \in \{1,2,4,8\}$, and the reliability study, which uses a single checkpoint to expose the per-target binding.

\begin{figure}[t]
\centering
\incfig{\columnwidth}{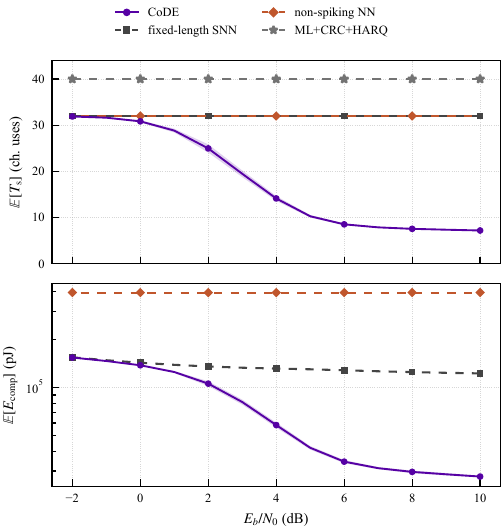}
\caption{Decision latency (top) and compute energy (bottom, log scale) versus $E_b/N_0$ on AWGN at $\epstarg=0.10$. On latency, CoDE commits before the deadline $\dline$; the fixed-length SNN and non-spiking NN sit at $\dline$, and ML+CRC+HARQ above it by its CRC overhead. Energy compares the neural backbones. Bands span the calibration draws.}
\label{fig:tradeoff}
\end{figure}

We compare CoDE against four baselines. The \emph{fixed-length SNN} runs our SNN to all $\dline$ symbols, with no early stop or certificate; the \emph{coverage-only} scheme stops on set size in the SpikeCP style and controls set coverage rather than undetected error, without the erasure or HARQ semantics; and the \emph{non-spiking NN} is a stateless dense feedforward neural network of comparable size that reads the whole codeword in one pass. The \emph{ML+CRC+HARQ} baseline is the standard error-control stack: full-length maximum-likelihood decoding with a $16$-bit CRC and HARQ retransmission. On this codebook ML decodes without error, so its CRC never fails and it operates at the deadline plus a CRC overhead of eight channel uses on QPSK. To compare across backbones we convert operations to energy with a stated $45$\,nm model~\cite{song2024xpikeformer}: the non-spiking NN's dense multiply-accumulate operations cost $4.6$\,pJ each, whereas a spiking receiver is multiplication-free and costs $0.9$\,pJ per accumulate. The resulting per-packet $\mathbb{E}[E_{\mathrm{comp}}]$ is reported in Fig.~\ref{fig:tradeoff} as a model proxy, not a silicon measurement.

\subsection{Reliability and Erasure}
\label{sec:exp-reliability}
Fig.~\ref{fig:reliability} shows $P_{\mathrm{ue}}$ and $P_{\mathrm{er}}$ against $E_b/N_0$, with a horizontal line at each $\epstarg$. Every $P_{\mathrm{ue}}$ curve stays at or below its target line, as guaranteed by Proposition~\ref{prop:ue}. $P_{\mathrm{ue}}$ peaks at intermediate SNR: the receiver mostly erases at low SNR and mostly decodes correctly at high SNR, so the bound binds in the middle regime. The measured rate sits below the target by a margin, as a marginal, distribution-free guarantee is conservative for a finite calibration set; $\epstarg$ is the operating knob a designer turns, and a CP-aware objective is the route to a tighter binding.

The erasure rate is the price of the certificate. It is high at low SNR and falls toward a floor near $\epstarg$ as SNR rises, and a looser target erases less.

\subsection{Latency, Energy, and the Spiking Backbone}
\label{sec:exp-tradeoff}
Fig.~\ref{fig:tradeoff} showcases the two resource savings of CoDE on AWGN at $\epstarg = 0.10$. On latency, the fixed-length SNN and the non-spiking NN sit at the deadline $\dline$, ML+CRC+HARQ sits above it by its CRC overhead, and CoDE commits earlier than all of them as $E_b/N_0$ rises. We compare energy only across the neural backbones, as the ML+CRC+HARQ decoder computes by correlation rather than neural operations and does not fit the same energy model. The non-spiking NN costs about $5\times$ the fixed-length SNN at every SNR, a fixed spiking-compute saving; CoDE then falls a further factor below the fixed-length SNN that is negligible at low SNR, where it still reads most of the packet, and grows to about $4\times$ by $10$\,dB, the adaptive-stopping saving. The two compound to roughly $5\times$ lower energy than the non-spiking NN at low SNR and nearly $30\times$ at $10$\,dB.

\begin{figure}[t]
\centering
\incfig{\columnwidth}{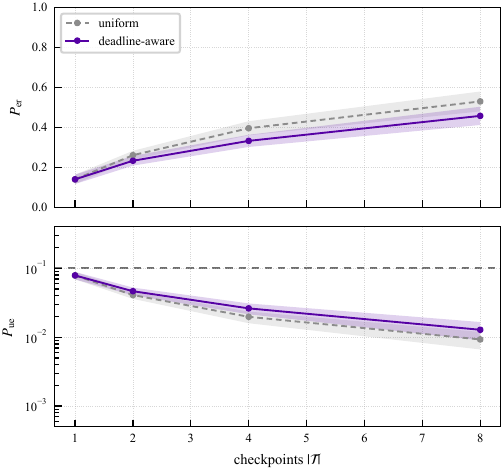}
\caption{Uniform Bonferroni versus deadline-aware budget split, as a function of the number of checkpoints $|\cps|$, at $\epstarg=0.10$ and $\dline=32$. Top: erasure rate. Bottom: undetected error rate.}
\label{fig:alloc}
\end{figure}

\subsection{Deadline-Aware Allocation}
\label{sec:exp-alloc}
Fig.~\ref{fig:alloc} isolates the budget allocation. At a fixed target and deadline we sweep $|\cps|$ and compare the uniform Bonferroni split with the deadline-aware split. As $|\cps|$ grows, the uniform split tightens every checkpoint and erases more, whereas the deadline-aware split concentrates the budget at the late checkpoints and erases less. At $|\cps| = 8$ and $\epstarg = 0.10$ the deadline-aware split cuts the erasure rate from $0.53$ to $0.46$, a gap that widens with $|\cps|$, while both splits hold $P_{\mathrm{ue}}$ near $0.01$, an order of magnitude under the target.

\subsection{Comparison with Baselines}
\label{sec:exp-baselines}
Fig.~\ref{fig:baselines} places CoDE against the baselines on the reliability-resource plane at a fixed SNR on AWGN. CoDE traces a certified frontier with $P_{\mathrm{ue}} \le \epstarg$ across the targets, a certificate none of the other schemes carries.
At $E_b/N_0 = 4$\,dB and $\epstarg = 0.05$, CoDE commits at a certified $P_{\mathrm{ue}} = 0.012$, a $48\%$ latency and $41\%$ compute reduction over the fixed-length SNN, whose uncertified $P_{\mathrm{ue}} \approx 0.04$ is higher despite reading the full packet. The coverage-only scheme stops earliest but its $P_{\mathrm{ue}} = 0.079$ exceeds the target, while the non-spiking NN and ML+CRC+HARQ reach $P_{\mathrm{ue}} \approx 0$ only at the full deadline; CoDE does not beat their accuracy but commits far earlier. At full length this codebook is comfortably decodable, so CoDE's undetected errors come from committing early on partial evidence, not a channel floor; the certificate is what keeps that commit safe.

\begin{figure}[t]
\centering
\incfig{\columnwidth}{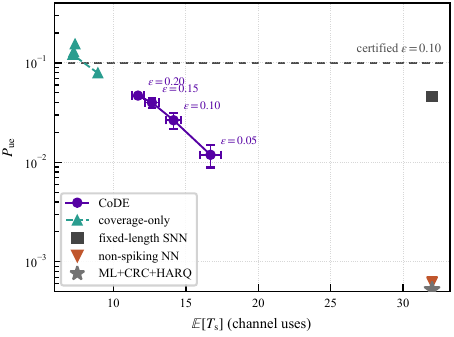}
\caption{CoDE versus baselines on the reliability-resource plane at $E_b/N_0=4$\,dB on AWGN, swept over $\epstarg \in \{0.05, 0.10, 0.15, 0.20\}$. Baselines: fixed-length SNN, coverage-only, non-spiking NN, and ML+CRC+HARQ; the error-free non-spiking NN and ML+CRC+HARQ points sit at the axis floor.}
\label{fig:baselines}
\end{figure}

\section{Conclusion}
\label{sec:concl}

We presented CoDE, a spiking receiver that decodes or erases short packets under a distribution-free reliability certificate: the conformal coverage level bounds the undetected error rate for any pretrained SNN and any calibration size. A single stopping rule sets reliability, latency, and both energy terms, and a deadline-aware error-budget allocation lowers the erasure rate at equal reliability. On AWGN at a $0.05$ target, CoDE matched a fixed-length decoder's reliability while roughly halving its latency and compute. 

\bibliographystyle{IEEEtran}
\bibliography{references}

\end{document}